\documentclass[12pt]{article}
\usepackage{latexsym}
\usepackage[dvips]{color}
\usepackage{amssymb}
\usepackage{graphicx}
\usepackage{epsfig}
\usepackage{amsmath}
\usepackage{mathrsfs}
\usepackage{natbib}
\usepackage{hyperref}
\usepackage{ulem}
\usepackage[utf8]{inputenc}

\usepackage{xcolor}

\renewcommand{\cite}[1]{\citet{#1}}

\newtheorem{theorem}{Theorem}[section]
\newtheorem{proposition}[theorem]{Proposition}
\newtheorem{corollary}[theorem]{Corollary}
\newtheorem{lemma}[theorem]{Lemma} 
\newtheorem{definition}[theorem]{Definition}
\newtheorem{remark}[theorem]{Remark} 
\newtheorem{example}[theorem]{Example}

\newcommand{\cE}{{\cal E}}
\newcommand{\cK}{{\cal K}}
\newcommand{\cR}{{\cal R}}

\newcommand{\cP}{{\cal P}}

\newcommand{\sigmaunten}{\underline{\sigma}}
\newcommand{\sigmaoben}{\overline{\sigma}}
\newcommand{\qv}[1]{\langle #1 \rangle}

\newcommand{\Einf}{\underline{\mathbb E}}

\newcommand{\EE}{\mathbb{E}}
\newcommand{\FF}{\mathbb{F}}
\newcommand{\HH}{\mathbb{H}}

\newcommand{\cF}{{\cal F}}

\newcommand{\E}{\EE}

\newcommand{\qed}{\mbox{ }~\hfill~$\Box$ \vspace{1ex} }

\newenvironment{proof}{\noindent{\sc Proof: }}{ \qed }

\newcommand{\rmII}{\text{\it I\kern-.08em I\,}}
\newcommand{\rmIII}{\text{\it I\kern-.08em I\kern-.08em I\,}}
\newcommand{\rmIV}{\text{\it I\kern-.08em V\,}}

\begin{document}

\title{Non--Implementability  of Arrow--Debreu Equilibria by Continuous Trading  under Knightian Uncertainty }
\author{
\sc Frank Riedel\footnote{Center for Mathematical Economics, Bielefeld University, 33501 Bielefeld,  Germany. Email: friedel@uni-bielefeld.de}
 \qquad Patrick Bei\ss ner\footnote{Corresponding Author. Center for Mathematical Economics, Bielefeld University, 33501 Bielefeld,  Germany. Email: patrick.beissner@uni-bielefeld.de}
 \\
\small Center for Mathematical Economics\\[-2pt] \small
Bielefeld University }
\date{\today}
\maketitle

\begin{abstract}
Under risk,  Arrow--Debreu equilibria can be implemented as Radner equilibria by continuous trading of few long--lived securities. We show that this result generically fails if there is Knightian uncertainty in the volatility. Implementation is only possible if all discounted net trades of the equilibrium allocation are mean ambiguity--free.
\end{abstract}

\medskip
{\footnotesize{ \it Key words and phrases:} Knightian Uncertainty, Ambiguity, general Equilibrium, Asset Pricing, Radner Equilibrium  \\
{\it \hspace*{0.6cm} JEL subject classification: D81, C61, G11} }

\newpage

\section{Introduction}

A celebrated and fundamental result of financial economics characterizes the situations in which all contingent consumption plans can be financed by trading  few long--lived assets; in diffusion models, asset markets are dynamically complete if the number of risky assets corresponds to the number of independent sources of uncertainty. When information is generated by a $d$--dimensional Brownian motion, $d$ risky assets consequently suffice to span a dynamically complete market.

In such a setting, one can thus expect an equivalence between the rather heroic equilibria of the Arrow--Debreu type -- where all trade takes place on a perfect market for contingent claims at time zero, and no trade ever takes place afterwards, --  and the more realistic Radner equilibra where agents trade long--lived assets dynamically over time.

Such equivalence of static and dynamic equilibria for diffusion models  has been established in different settings and at different levels of generality by \cite{DuffieHuang85}, \cite{DuffieZame89}, \cite{KaratzasEtAl90}, \cite{AndersonRaimondo08}, \cite{RiedelHerzberg13}, \cite{HugonnierMalamudTrubowitz12}.

In this paper, we show that this celebrated equivalence generically breaks down under Knightian uncertainty about volatility. 
 We place ourselves in a framework which makes it as easy as possible for the market to span the equilibrium allocations. Even then, we claim, Arrow--Debreu equilibria will  usually  not be implementable by a dynamic market if there is Knightian uncertainty in individual endowments.
 
 In which sense do we make it easy? First, we consider a model in which a $d$--dimensional Brownian motion with ambiguous volatility generates the economy's information flow. Second, as in the Duffie--Huang--approach, we consider nominal asset markets. The nominal asset structure allows for an exogenously chosen asset structure. If there is no spanning in this setting, one cannot expect spanning in the more demanding real asset setting considered by  \cite{AndersonRaimondo08}, \cite{RiedelHerzberg13}, and \cite{HugonnierMalamudTrubowitz12} where security prices and consumption prices are endogenously determined in equilibrium and linked via the real dividend structure. 
Third, we consider a setting where aggregate endowment is ambiguity--free. This is the ideal starting point for an economic analysis of insurance properties of competitive markets. In this setting, a ``good'' economic institution should lead to an ambiguity--free allocation for all (ambiguity--averse) individuals. 
Indeed, we show that the efficient (and thus, Arrow--Debreu equilibrium) allocations in this Knightian economy provide full insurance against uncertainty. This generalizes the results of \cite{BillotChateauneufGilboaTallon00}, \cite{Dana02}, \cite{Tallon98} and \cite{de2011ambiguity}
to the continuous--time setting with non--dominated sets of priors.

We  study the possibility of implementation  in the so--called Bachelier model where the risky (or, in this Knigtian setting, maybe better: uncertain) asset is given by the Brownian motion itself because  this case is particularly transparent. Indeed, in the classic case, one can then immediately apply the martingale representation theorem to find the portfolio strategies that finance the Arrow--Debreu (net) consumption plans. 
We study under what condition this result holds in an uncertain world. Under Knightian uncertainty about volatility, the martingale representation theorem changes in several aspects. Implementation is possible if and only if the value of net trades is mean ambiguity--free, or in other words, if the expected value of net trades is the same for all priors.  

We thus completely clarify under what conditions one can implement Arrow--Debreu as Radner equilibria. Clearly, being free of  ambiguity in the mean   is  weaker than being free of ambiguity  in the strong sense of having the same probability distribution under all priors. Nevertheless, our result implies that ``generically'', implementation will be impossible under Knightian uncertainty about volatility. We show this explicitly in the case when there is no aggregate uncertainty in the economy. The set of all endowments for which implementation fails  is prevalent.\footnote{In infinite--dimensional settings, there is no obvious notion of genericity. We use the concept of prevalence (or shyness) introduced into theoretical economics by \cite{AndersonZame01} who show that it is a reasonable measure--theoretic generalization to infinite--dimensional spaces of the usual ``almost everywhere''--concept in finite--dimensional contexts. Compare also \cite{HuntSauerYorke92}.}


An additional contribution of our paper concerns  the existence of Arrow--Debreu equilibria; without existence, the question of implementation would be void. Existence is not a trivial application of the well--known results on existence of general equilibrium for Banach lattices as the informed reader might think (and the authors used to think as well). Under volatility  uncertainty, the natural  commodity space  combines the well--known $L^p$--space with some degree of continuity. In fact, the right commodity space   consists of contingent claims that are suitably integrable or even bounded almost surely for all priors, and are \textit{quasi--continuous}. A  mapping is quasi--continuous if it is continuous in nearly all its domain.  The property of quasi--continuity comes for free in the probabilistic setting: Lusin's theorem establishes the fact that any measurable function on a nice topological space is   quasi--continuous. Under volatility uncertainty, this equivalence between measurability and quasi--continuity no longer holds true. We are thus led to study a new commodity space which has not been studied  so far in general equilibrium theory. Compare also the discussion of this space in the recent papers  \cite{EpsteinJi13},  \cite{EpsteinJi14} and \cite{Equilibrium2014Beissner}.

For this commodity space, the available existence  theorems do not immediately apply. The abstract question of existence must thus be dealt with separately, but  we leave the general question of existence for the future as it is not the main concern of this paper. In this paper, we use a different approach to establish existence. In our homogenous  framework, one can show that the efficient allocations coincide with the efficient allocations under risk (compare the related results of \cite{Dana02} in a static setting). When all agents share the same prior, it is well known that the efficient allocations are independent of the prior and can be determined by maximizing a suitable weighted sum of utilities pointwise. As a consequence, the efficient allocations are continuous functions of aggregate endowment;  they are quasi--continuous if aggregate endowment is quasi--continuous. This allows us to establish existence in our new commodity space where quasi--continuity is required. One can just fix any prior $P$ and choose an Arrow--Debreu equilibrium in the expected utility economy where all agents use this prior.  These equilibria are also equilibria under Knightian uncertainty. As a by--product, we obtain indeterminacy of equilibria, as in related Knightian settings, such as \cite{Tallon98}, \cite{Dana02}, \cite{RigottiShannon05}, or \cite{DanaRiedel13}.

The paper is set up as follows. The next section describes the economy. Section 3 studies efficient allocations and Arrow--Debreu equilibria. Section 4 contains our main results on (generic non--)implementability. { Section 5 concludes.} The appendix contains additional material on Knightian uncertainty in continuous time.

\section{\!The Economy under Knightian Uncertainty}

We consider an economy over the time interval $[0,T]$ with Knightian uncertainty.

There is a  finite set of agents $\mathbb{I}=\{1,\ldots,I\}$ agents in the economy who care only about consumption at terminal time $T$. The agents share a common view of the uncertainty in their world in the sense that they all agree on  the same set of  priors $\cP$ on the state space $\Omega=C[0,T]^d$ of possible trajectories for  the canonical $d$--dimensional Brownian motion $W_t(\omega)=\omega(t)$. 


 Quite general specifications are possible here, and we discuss some of them  in the last section. To have a concrete model, we start here with volatility uncertainty for $d$ independent Brownian motions. Let $\Sigma=\prod^d_{k=1} [\underline \sigma^k, \overline \sigma^k]$ for $0<\underline \sigma^k \le \overline \sigma^k, k=1,\ldots,d$.  $\Sigma$ models the possible values of the volatilities of our ambiguous  Brownian motion $W$. The set of priors $\cP$ consists of all probability measures $P$ that make $W$ a martingale such that the  covariation between any $W^k$ and $W^l$ satisfies $\left(\langle W^k, W^l \rangle_t\right)_{k,l=1,\ldots,d} \in \Sigma $ for all $t\ge 0$ $P$--a.s.
In our concrete case, this means that the covariation between two different Brownian motions vanishes and the variation of a Brownian motion $W^k$ satisfies $ \left( \underline \sigma^k\right)^2 t  \le \langle W^k \rangle_t\le \left( \overline\sigma^k \right)^2 t$ for all $t \ge 0$.

 The set of priors is not dominated by a single probability measure. In such a context, sets that are conceived as null by the agents cannot be identified with null set under one probability measure. Indeed, as possible scenarios are described by  a whole class of potentially singular priors, an event can only be considered as negligible or null when it is a null sets under all priors simultaneously. Such sets are called polar sets; the corresponding sure events, those that have probability one under all priors, are called \textit{quasi--sure} events.

These issues require a reconsideration of some measure theoretic results. Under risk, a measurable function is ``almost'' continuous in the sense that for every $\epsilon>0$ there is an open set $O$ with probability at least $1-\epsilon$  such that the function is continuous on $O$; this is Lusin's theorem. Under non--dominated Knightian uncertainty, this Lusin property, or quasi--continuity, does not come for free from measurability, and one needs to impose it. We refer to  \cite{EpsteinJi13} and \cite{DenisHuPeng11} for the financial and measure--theoretic background.

 A proper commodity space under non--dominated Knightian uncertainty then consists of all bounded quasi--continuous functions; it is denoted by   $\HH=L^\infty_\cP$.  The boundedness assumption of endowments  is made for ease of exposition and to keep the arguments as concise as possible for our aims; it can be relaxed, of course, as we discuss later on. The consumption set, denoted by $\HH_+$, consists of quasi-surely positive functions in $\HH$.

Agents' preferences are given by  Gilboa--Schmeidler--type expected utility functionals of the form
$$ U^i(c)= \Einf u^i(c) = \inf_{P \in \cP} E^P u^i(c) $$
for nonnegative consumption plans $c \in \HH_+$ and a Bernoulli utility function
$$ u^i : [0,\infty) \to \mathbb R$$
 which is concave, strictly increasing, sufficiently smooth, and satisfies the Inada condition 
$$\lim_{x \downarrow 0} \frac{\partial u^i}{\partial x} (x)=\infty\,.$$

In the following, we will denote by $\cE$ the economy we just described and refer to it as the \textit{Knightian economy} to distinguish it from expected utility economies $\cE^P$ later on. In the economy $\cE^P$, agents use expected utility $E^P u^i(c)$ for a particular common prior $P \in \cP$; otherwise, the economy has the same structure as $\cE$.

Agents have an endowment $e^i$ which is bounded and bounded away from zero quasi--surely.

In the following, we consider a situation where the market can potentially insure individuals against their idiosyncratic ambiguity because ambiguity washes out in the aggregate. We explicitly do allow for aggregate (and individual) risk. 

\begin{definition}\label{DefNoamb}
$X \in \HH$ is ambiguity--free (in the strong sense) if $X$ has the same probability distribution under all priors $P,Q \in \cP$.



\end{definition}

We assume throughout this paper that
aggregate endowment $e$ is ambiguity--free in the strong sense.
We thus make it as easy as possible for the market to provide insurance.  

\section{Efficient Allocations and Arrow--Debreu Equilibria with No Aggregate Ambiguity}

Before considering the possibility or impossibility of implementing Arrow--Debreu equilibria by trading  few long-lived assets, we need to study existence and properties of equilibria in our context. 
 It will turn out  that efficient allocations, and a fortiori Arrow--Debreu equilibrium allocations are ambiguity--free. If we allow for a complete set of forward markets at time zero, the market is thus able to ensure all individuals against their idiosyncratic Knightian uncertainty.

 We recall the structure of efficient allocations in the homogenous expected utility case. If all agents agree on one particular probability $P$, the efficient allocations are independent of that $P$; indeed, they are characterized by the equality of marginal rates of substitution. For weights $\alpha^i \ge 0$, the efficient allocation $c_\alpha$ maximizes pointwise the sum 
$$\sum_{i\in \mathbb{I}} \alpha^i  u^i(c^i)$$ over all vectors $c \in \mathbb R^I_+$ with $\sum_{i\in \mathbb{I}} c^i \le e(\omega)$. It  is characterized by the first--order conditions
$$ \alpha^i \frac{\partial u^i}{\partial x} (c_{\alpha}^i)=
\alpha^j \frac{\partial u^j}{\partial x} (c_{\alpha}^j)$$ for all agents with strictly positive weights $\alpha^i, \alpha^j>0$; the agents with   weight zero consume zero, of course.
Due to our assumptions, one can then write the efficient allocations as a continuous function
$$c_\alpha=c_\alpha(e)$$
of aggregate endowment. In the following, we denote by $\Delta$ the simplex of weights $\alpha$ that sum up to $1$ and by $\mathbb O=(c_\alpha)_{\alpha\in \Delta}$ the set of efficient allocations in homogenous expected utility economies.

\begin{theorem}\label{ThmEfficient}
\begin{enumerate}
\item Every efficient allocation in $\cE$ is ambiguity--free.
\item The efficient allocations in the Knightian economy $\cE$ coincide with the efficient allocations $(c_\alpha)_{\alpha \in\Delta}$ in homogenous expected utility economies $\cE^P$ and are independent of a particular prior $P \in \cP$.
\end{enumerate}
\end{theorem}

\begin{proof}
Note that the efficient allocations in homogenous expected utility economies $\mathbb O$ are ambiguity--free as they can be written as monotone functions of aggregate endowment. As these functions are continuous, they are also quasi--continuous in the state variable and thus do belong to our commodity space.

We first show that these allocations are also efficient in our Knightian economy. 
Fix some weights $\alpha \in \Delta$ and fix some $P \in \cP$. Let $c$ be a feasible allocation.
Then we have
\begin{align*}
\sum_{i\in \mathbb{I}} \alpha^i U^i(c^i) &\le 
\sum_{i\in \mathbb{I}} \alpha^i E^P u^i(c^i)
= E^P \sum_{i\in \mathbb{I}} \alpha^i  u^i(c^i) \,.
\end{align*}
As $c_\alpha$ maximizes the weighted sum of utilities pointwise, we continue with
$$ E^P \sum_{i\in \mathbb{I}}\alpha^i  u^i(c^i)  \le  E^P \sum_{i\in \mathbb{I}} \alpha^i  u^i(c^i_\alpha)
=  \sum_{i\in \mathbb{I}} \alpha^i  E^P  u^i(c^i_\alpha)\,.$$
Now $c_\alpha$ is ambiguity--free, hence we have
$ E^P  u^i(c^i_\alpha)= E^Q u^i(c^i_\alpha)$ for all $P,Q \in \cP$ and thus
$ E^P  u^i(c^i_\alpha)= \Einf   u^i(c^i_\alpha) = U^i(c^i_\alpha)$. We conclude that $c_\alpha$ maximizes the weighted sum of Gilboa--Schmeidler utilities. It is thus efficient.

Now let $d$ be another efficient allocation. By the usual separation theorem, it maximizes the weighted  sum of utilities $$\sum_{i\in \mathbb{I}} \alpha^i U^i(c^i)  $$ for some $\alpha \in \Delta$. 
Set
$$ \Gamma = \left\{ \omega \in \Omega : \mbox{there is $i \in \mathbb I$ with $d^i(\omega) \not= c^i_\alpha(\omega)$}\right\}\,.$$ Due to our strict concavity assumptions, we have
$$ \sum_{i\in \mathbb{I}} \alpha^i u^i(d^i(\omega)) < \sum_{i\in \mathbb{I}} \alpha^i u^i(c_\alpha^i(\omega)) $$ for all $\omega \in \Gamma$. Now assume that $\Gamma$ is not a polar set. Then there is $P \in \cP$ with $P(\Gamma)>0$. Therefore, we have
$$ E^P \sum_{i\in \mathbb{I}} \alpha^i u^i(d^i) < E^P \sum_{i\in \mathbb{I}} \alpha^i u^i(c_\alpha^i) \,.$$
Since $c_\alpha$ is ambiguity--free, $E^P \sum_{i\in \mathbb{I}} \alpha^i u^i(c_\alpha^i) = \sum_{i\in \mathbb{I}} \alpha^i U^i(c_\alpha^i)$.
On the other hand, by ambiguity--aversion, $$ \sum_{i\in \mathbb{I}} \alpha^i U^i(c_\alpha^i)\le E^P \sum_{i\in \mathbb{I}} \alpha^i u^i(d^i)\,,$$ and we obtain a contradiction. We thus conclude that $\Gamma$ is a polar set and thus $d=c_\alpha$ quasi--surely.
\end{proof}

The preceding  theorem obtains the same characterization of efficient allocations in our Knightian case as  \cite{Dana02} does for Choquet expected utility economies. The argument is different, though: as our set of priors does not lead to a convex capacity, one cannot use  the comonotonicity of efficient allocations  to identify a unique worst--case measure for the agents.  Instead, we rely on aggregate endowment being ambiguity--free to reach the  conclusion that all efficient allocations are ambiguity--free and coincide with the efficient allocations of any expected utility economy in which all agents share the same prior. Moreover, we do not have a dominating measure here. 

As a consequence of the theorem we obtain the following generalization of \cite{BillotChateauneufGilboaTallon00} to our non--dominated setting.

\begin{corollary}\label{EffCor}
If there is no aggregate uncertainty, i.e. $e\in \mathbb R_+$ quasi--surely, then all efficient allocations are full insurance allocations.
\end{corollary}

Let us now turn to Arrow--Debreu equilibria. In the first step, it is important to clarify the structure of price functionals. A price is a positive and continuous linear price functional $\Psi$ on the commodity space
$L^\infty_\cP$. 
 A typical representative of these price functionals has the form of a pair $(\psi,P)$ where $\psi$ is the \textit{state--price density} as in the expected utility case and $P$ is some particular prior in $\cP$. Knightian uncertainty adds the new feature  that the market also chooses the measure $P\in\cP$.

An \textit{Arrow--Debreu equilibrium} consists then of a feasible allocation $c=(c^i)_{i \in \mathbb{I}}\in \mathbb H_+^{\mathbb I}$ and a price $\Psi$ such that for all $d \in \mathbb H_+$ 
the strict inequality 
$ U^i(d) > U^i(c^i) $ implies $\Psi(d)>\Psi(e^i)$.

We are going to show that Arrow--Debreu equilibria exist and equilibrium prices are indeterminate.
Recall that $\cE^P$ is the economy in which agents have standard expected utility preferences $U^i_P(c)=E^P u^i(c)$ for the same prior $P \in \cP$ and the same endowments as in our original economy. 
Bewley (Theorem 2 in \cite{Bewley72}) has shown that Arrow--Debreu equilibria exist with state--price densities in $L^1(\Omega,\mathcal{F},P)$ for the economy $\cE^P$. By the first welfare theorem, an equilibrium allocation in $\cE^P$ can be identified with some $c_\alpha$ for  a vector of weights $\alpha\in\Delta$ and the corresponding equilibrium state--price density with 
$\psi_\alpha= \alpha^i   \left(u^i\right)' ( c^i_\alpha)$. Due to our assumption
that endowments are in the interior of the consumption set, all weights $\alpha_i$ are strictly positive.

We are now ready to characterize the equilibria of our Knightian economy.

\begin{theorem}\label{ThmEquilibrium}
\begin{enumerate}
\item
 Let $(c_\alpha,\psi_\alpha)$ be an equilibrium of $\cE^P$ for some $\alpha \in \Delta$ and $P \in \cP$. Then $(c_\alpha, \psi_\alpha \cdot P)$ is an Arrow--Debreu equilibrium of the Knightian  economy.
\item
If $(c,\Psi)$ is an Arrow--Debreu equilibrium of $\cE$, then there exists $\alpha\in \operatorname{rint} A$  and $P \in \cP$ such that $c=c_\alpha$ and $\Psi= \psi_\alpha \cdot P$ with  $$\psi_\alpha= \alpha^i   \left(u^i\right)' ( c^i_\alpha)\quad \textit{ for } i \in\mathbb I.$$ 
\end{enumerate}

In particular, Arrow--Debreu equilibria 
\begin{itemize}
\item exist,
\item their price is indeterminate,
\item and the allocation is ambiguity--free.
\end{itemize}
\end{theorem}

\begin{proof}
Let $(c_\alpha,\psi_\alpha)$ be an equilibrium of $\cE^P$. $c_\alpha$ obviously clears the market and is budget--feasible in the Knightian economy because we use the same pricing functional as in the economy $\cE^P$.
It remains to show that $c^i_\alpha$ maximizes utility in the Knightian economy subject to the budget constraint. 

In the first place, we need to verify that $c_\alpha$ belongs to our commodity space $\mathbb H$ (which is smaller than the commodity space $L^\infty(\Omega, \cF, P)$ considered by  \cite{Bewley72} as it contains only quasi--continuous elements).
But we have already noted above that the $c_\alpha$ are quasi--continuous, and thus elements of $\mathbb H_+$, because they can be written as continuous functions of $e$.

 Let $d$ be budget--feasible for agent $i$. As $c_\alpha$ is an Arrow--Debreu equilibrium in the expected utility economy $\cE^P$, we have  $E^P u^i(c_\alpha^i) \ge E^P u^i(d)$. 
As $c_\alpha^i$ is ambiguity--free by Theorem \ref{ThmEfficient}, we have  $ U^i(c_\alpha^i)=E^P u^i(c_\alpha^i)$. 
Therefore,
$$U^i(d) \le E^P u^i(d) \le E^P u^i(c_\alpha^i) = U^i(c_\alpha^i)$$
and we are done.

For the converse, let $(c,\Psi)$ be an Arrow--Debreu equilibrium of $\cE$. By the first welfare theorem and Theorem \ref{ThmEfficient}, there exist $\alpha\in \Delta$ with $c=c_{\alpha}$. All $\alpha^i>0$ as individual endowments are strictly positive. (Otherwise, $c^i_\alpha=0$ which is dominated by $e^i$.) Due to utility maximization, $\Psi$ has to be a supergradient of $U^i$ at $c^i_\alpha$.
 The set of  supergradients consists of linear functionals of the form $\psi_\alpha \cdot \cP$, where $P$ is a minimizer in the set of priors. Since $u^i(c^i_\alpha)$ is ambiguity--free,  $E^Pu^i(c^i_\alpha) $ is constant on $\mathcal{P}$ and hence every element in $\mathcal{P}$ is a minimizer of the Gilba--Schmeidler--type expected utility.
\end{proof}

\section{(Non--)Implementability by Continuous Trading of Few Long--Lived Securities}

We now tackle the question if the efficient allocations of Arrow--Debreu equilibria can be implemented by   trading a few long--lived assets dynamically over time. Under risk, the answer is (essentially) affirmative. If we allow the market to select the asset structure,  \cite{DuffieHuang85} establish  implementability. In this case, we have purely nominal assets whose dividends are not directly related to commodities. One can thus choose their prices independently of consumption prices.  In general, of course, the asset price structure with real assets is  endogenous. In that case, the question of Radner implementability is much more complex and was only recently solved by  \cite{AndersonRaimondo08}, \cite{RiedelHerzberg13} and \cite{HugonnierMalamudTrubowitz12}.  If the asset market is \textit{potentially complete} in the sense that sufficiently many independent dividend streams are traded, then one can obtain endogenously dynamically complete asset markets in sufficiently smooth Markovian economies. For non--smooth economies and non--Markovian state variables, the question is still open.

As we focus on the intrinisic limit of  implementability which is created by Knightian uncertainty, we make here the life as easy as possible for the financial market: as in \cite{DuffieHuang85}, we consider the case with nominal assets freely chosen by the market. Since the equivalence between Arrow--Debreu and Radner equilibria usually breaks down, the result is stronger if we allow the market to choose the asset structure for nominal assets. If one cannot even implement the Arrow--Debreu equilibrium in the nominal case, one cannot do so with real assets either.

\paragraph{The Bachelier Market}
We consider first the simplest case of a so--called Bachelier market. 
There is a riskless asset $S^0_t=1$. Moreover, the price of the other $d$ assets is given by our $d$-dimensional ambiguous (or $\mathbb{E}$)--Brownian motion $S_t=B_t$. 
A  trading strategy  then consists of a process $(\theta^1,\ldots,\theta^d)=\theta \in\Theta(S) $, the space of admissible integrands for (Knightian) Brownian motion  (see Appendix A.2, \cite{EpsteinJi13} or \cite{DenisHuPeng11} on details for the admissible integrands). The agents start with wealth zero (as there is no endowment nor consumption at time zero). Their gains from trade at time $t$ are  then
\begin{equation}\label{1}
G^\theta_t= \theta_t S_t=  \int_0^t   \theta_u\textnormal{d}S_u=\sum_{1\leq k\leq d}   \int_0^t \theta^k_u  \textnormal{d}S^k_u\,.
\end{equation}
 They can afford to consume $c^i\in \mathbb H_+ $ with $(c^i-e^i) \psi=G^{\theta^i}_T$ where $\psi$ is the spot consumption price at time $T$, a nonnegative $\cF_T$--measurable function.
A budget--feasible consumption--portfolio plan $(c^i, \theta^i)$ is a pair of a consumption plan and a trading strategy that satisfy the above budget constraint.

 A \textit{Radner equilibrium} consist  of a spot consumption price $\psi$ and portfolio--consumption plans  $(c^i,\theta^i)_{i\in\mathbb{I}}$ such that markets clear, i.e. $\sum c^i = e$, $\sum \theta^i = 0$, and agents maximize their utility over all portfolio--consumption plans that satisfy the budget constraint.

\begin{theorem}\label{ThmRad}
In the Bachelier model with asset prices $S=B$, an Arrow--Debreu equilibrium of the form $(c,   \psi \cdot P)$ can be implemented as a Radner equilibrium if and only if the value of the net trades $\xi^i=\psi(c^i-e^i)$ are \textit{mean ambiguity--free}, i.e. for all $P,Q \in \cP$
$$E^P \xi^i= E^Q \xi^i\,.$$ 
\end{theorem}

\begin{proof}
 Let $(c, \psi \cdot P) $ be an Arrow-Debreu Equilibrium as in Theorem \ref{ThmEquilibrium}. 
 
 Suppose we have an implementation in the Bachelier model with trading strategies $\theta^i$.  
The Radner budget constraint gives 
$$\xi^i= G^{\theta^i}_T= \int_0^T \theta^i_t \textnormal{d}B_t \quad \mathcal{P }\textnormal{-quasi surely}.$$ Stochastic integrals are symmetric $\mathbb{E}$--martingales, i.e.
$E^P \int_0^T \theta^i_t \textnormal{d}B_t = \mathbb E \int_0^T \theta^i_t \textnormal{d}B_t$ for all $P\in\mathcal{P}$.
Consequently the value of each net trade is mean ambiguity--free.

We show now that implementation is possible if net trades are mean ambiguity--free.
We divide the proof into three steps.
First, we introduce  the candidate trading strategies for agent $i\neq I$ and show  market clearing in the second step. Finally we  show  that these strategies are maximal  in the budget sets. 

 Let the value of the net trade $\xi^i$  be mean ambiguity--free for all agents $i\in\mathbb I$. The Arrow--Debreu budget constraint gives   $0= E^P\xi^i= \mathbb{E}\xi^i$. Consequently by Corollary  \ref{MRTCor}, the process $t\mapsto\mathbb{E}_t\xi^i$ is a symmetric $\EE$-martingale, so by    the martingale representation theorem \ref{MRT}
\begin{eqnarray*}
\mathbb{E}_t\xi^i= \int_0^t \theta_r^{i}  \textnormal{d} S_r,\quad \mathcal{P }\textnormal{-quasi surely}
\end{eqnarray*}
and
$$\xi^i = \int_0^T \theta^i_t \textnormal{d}S_t \,.$$
Hence, the trading strategies $\theta^i$ are candidates for trading strategies in a Radner equilibrium with allocation $c$ and spot consumption price $\psi$. 

By market--clearing in an Arrow--Debreu equilibrium, we have
$$0 = \sum_{i\in\mathbb I} \xi^i = \int_0^T  \sum_{i\in\mathbb I}  \theta^i_t 
\textnormal{d}B_t \,.$$
As stochastic integrals that are zero have a zero integrand (even under Knightian uncertainty; see Proposition 3.3 in \cite{soner2011martingale}), we conclude that the portfolios clear, $\sum_{i\in\mathbb I} \theta^i = 0$.

It remains to check that the consumption--portfolio strategy $(c^i,\theta^i)$ is optimal for agent $i$ under the Radner--budget constraint. Suppose there is a trading strategy $(d,\eta)$ with 
 $$ \psi (d-e^i)= \int_0^T \eta_t \textnormal{d}S_t \,.$$
We then have $$E^P \psi(d-e^i) = \mathbb E\int_0^T \eta_t \textnormal{d}S_t =  0\,,$$
and $d$ is thus budget--feasible in the Arrow--Debreu model. 
We conclude
that $U^i(d) \le U^i(c^i)$. 
\end{proof}

To illustrate the theorem,  we consider a concrete example.

\begin{example}
Let $e(\omega)\equiv 1$, $I=2$, $d=1$, and $u^i=\log$ for $i=1,2$. Assume $e^1=\phi(B_T)=\exp(B_T)\wedge 1$ and $e^2 =1-e^1$.  By Corollary \ref{EffCor} and Theorem \ref{ThmEquilibrium}, equilibrium allocations are full insurance; therefore, the  state price density is deterministic, so that we may take  without loss of generality $\psi =\frac{\alpha_i}{c_\alpha^i} =1$. 

In this case, the expected value of net trades $\xi^i=c^i-e^i$ depends on the particular measure $P \in \cP$. For example, if $B$ has constant volatility $\sigma$ under $P$, then
$$E^P \exp(B_T)\wedge 1 = \exp(\sigma^2/2) \Phi(-\sigma) + 1/2 $$
where $\Phi$ is the standard normal distribution. 
Radner implementation is therefore impossible.
\end{example}


The previous example suggests that Radner implementation might not be expected, in general. In the next step, we clarify this question in a world without aggregate uncertainty. We know from our analysis in the previous section that all Arrow--Debreu equilibria fully insure all agents in such a setting. We claim that ``for almost all'' economies, or ``generically'', Radner implementation is impossible.

While the notion of ``almost all'' has a natural meaning in the finite--dimensional context because one can use Lebesgue measure to define negligible sets as null sets under that measure, the notion of ``almost all'' does not generalize immediately to infinite--dimensional Banach spaces because there is no translation--invariant measure that assigns positive measure to all open sets on such spaces. 
\cite{AndersonZame01} develop the notion of prevalence which coincides with the usual notion of full Lebesgue measure in finite--dimensional contexts and is thus an appropriate generalization to infinite--dimensional settings. 

In the following, we fix aggregate endowment with no uncertainty $e >0$ and consider the class of economies parametrized by individual endowments
$$ \cK = \left\{ (e^i)_{i \in \mathbb I} \in \mathbb H_+^{\mathbb I} : 
\sum_{i\in\mathbb I} e^i = e \right\} \,.$$
We say that an economy with endowments $(e^i)_{i \in \mathbb I}$ does not allow for implementation if there is no Arrow--Debreu equilibrium $(c, \psi \cdot P)$ which can be implemented as a Radner equilibrium. Let $\cR$ be the subset of economies in $\cK$ which do not allow for implementation.

In Theorem \ref{ThmRad}, we introduced the notion mean ambiguity--free random variables in $\mathbb{H}$. The collection of all such elements is denoted by $\mathbb{M}$ and set  $\mathbb{M}^c= \mathbb{H}\setminus  \mathbb{M}$. For an alternative characterization of $\mathbb{M}$, see Corollary \ref{MRTCor}.

\begin{theorem}\label{ThmShy}
With no aggregate uncertainty,  implementation  of  Arrow--Debreu equilibrium  via a Bachelier Model is  generically impossible: $\cR$ is prevalent in $\cK$.
\end{theorem}

\begin{proof} Let $(c_\alpha, \psi\cdot P)$ be an equilibrium and let $e=1$ without loss of generality. By Corollary \ref{EffCor},   $e^i$ is the only non--constant within the value of  each net trade $\xi^i=\alpha_i  (u^i)'\left(c^i_\alpha \right)  \left(c^i_\alpha-e^i\right)$.  Consequently by Theorem \ref{ThmRad}  implementability fails if there is a $i\in\mathbb{I}$ with  $e^i\in \mathbb{M}^c$. 

From this observation, we  may  focus  on the prevalence of the property   ``mean ambiguity--free" within   the space of initial endowments $\mathcal{K}$. The following  claim is crucial; its proof is below. The order relation  $\leq$  induced by the cone $\HH_+$.
\\[.5em]
\textit{Claim:} $\mathbb{M}^c \cap[0,1]$  is a  prevalent set in $[0,1]=\{Y\in \mathbb{H}: 0\leq Y\leq 1\}$ of $\mathbb{H}$.
\\[0.5em]
In the case $I= 2$,   one  endowment within $[0,1]$  determines a distribution of endowments via $e^2=e-e^1$, hence
 $   \mathcal{K}\cap \mathbb{M}^\mathbb{I}$ is a  shy  subset of $ \mathcal{K}$. Therefore the claim implies that  $\mathcal{R}$ is prevalent in $\mathcal{K}$.

 For an arbitrary $\mathbb{I}$, we have $ (\mathbb{M}^c\cap [0,1])^\mathbb{I}$ is  a prevalent subset in $[0,1]^\mathbb{I}$ of $\mathbb{H}^\mathbb{I}$. This follow by the the same arguments as in the proof of the claim, by choosing the  subspace $T^\mathbb{I}$ of $\mathbb{H}^\mathbb{I}$  as the finite dimensional test space.
Hence, $\mathbb{M}^c\cap [0,1]\times (\mathbb{M}\cap [0,1])^{\mathbb{I}-1}$ is a shy set in $ [0,1]^\mathbb{I}$.  
The result then follows by an analogous argument as in the case $I=2$. 

We come now to the proof of the claim made above.
 The proof  relies on the Martingale representation Theorem \ref{MRT}. The  Corollary \ref{MRTCor} implies $M^1:=\mathbb{M} \cap[0,1]= \{(\eta, 0)\in \mathcal{M}\times \{0\}: \int_0^T \eta_t\textnormal{d}B_t \in [0,1]\}$,  where $ \mathcal{M}$ is the completion of piecewise constant process, see (\ref{Theta}) in Appendix A.3.

Let  $\mathbb{BV}$ denote the Banach space of progressive--measurable processes with continuous paths and of  bounded variation on $[0,T]\times \Omega$, so that   $$(\mathcal{M} \times \mathbb{K}_0)\cap [0,1]=:M\!K^1\subset X:=  (\mathcal{M}\times \mathbb{BV})\,,$$where $\mathbb{K}_0$ denotes all processes $  x+\hat K$ such  that $x\in \mathbb{R}$ and  $\hat K\in\mathbb{K}^M$.\footnote{Akin to (\ref{K}), set  
$
\mathbb{K}^M\!\!=\!\big\{
\!\!-K\!:\E\textit{--martingale},  K_0=0,  \textit{continuous, incr., }\: \mathbb{E}\sup_{t} K_t^2 <\infty \big\}$.}

The main step is to define a tractable ``test--space"  $T$  (to check prevalence) via  a concrete  $ K$: By  Remark \ref{MRTRem} and  setting $\varphi_t\equiv 1\in\mathbb{R}^d$, fix $K^1\in \mathbb{K}^M$ given by
\begin{eqnarray*}
 K^1_t= \int_0^t 1 \textnormal{d}\langle B\rangle_r - \int_0^t G( 1) \textnormal{d}r
=   \sum_{k=1}^d \langle B^k\rangle_t - \frac{t}{2}  \left(\overline \sigma^k\right)^2\, .
\end{eqnarray*}
The process $(-K^*_t)=-(x+ K^1_t)\in \mathbb{K}_0$ is  again an increasing $\EE$--martingale since $\EE$ preserves constants.  Positive homogeneity of $\EE$ implies that  $-aK^*$ is an $\EE$-martingale if $a\geq 0$.  A  rescaling of  $ K^*$  into $K$ yields $ K\in  \mathbb{K}_0\cap [0,1]$.\footnote{This is possible  since  $\min_k \left(\sigmaunten^k\right)^2 t \le \langle B^k\rangle_t \le \max_k \left(\sigmaoben^d\right)^2 t$ as discussed in the Appendix for the the case $d=1$.}

Fix the following one dimensional test space  $$T=\Big\{(0,   -aK): a\in \mathbb{R}\Big\} $$  of $X$ and denote the Lebesgue measure on $T$ by $\lambda_T$;  we have to check:
\begin{enumerate}
\item \textit{There is a $ c\in X$ with $\lambda_T(M\!K^1+c)>0$}:

  The arbitrary translation of $M\!K^1$ is performed by $c=(0,0)\in X $.
Since  only a positive $a\in\mathbb{R}$ makes $-aK$  an $\E$--martingale, we derive   
   $\lambda_T(M\!K^1)\geq\lambda_T( (0,  -aK): a\in [0,1])=\lambda_\mathbb{R}([0,1])=1>0$.

\item \textit{For all $ z\in X$ we have $\lambda_T(E+z)=0$}: 

This follows directly from  the definition of $T$ and $M^1$, since  at most one $K\in M\!K^1$ lies in  $M^1 +z$. The condition follows, since  $\lambda_T(\{0, K\})=0$ for every $K\in \mathbb{K}$.
\end{enumerate}
By Fact 6 of  \cite{AndersonZame01} every finitely shy set in $ M\!K^1$ is  also a shy set in $ M\!K^1$, and therefore $ \mathbb{M}^c\cap [0,1]= \mathbb{H}\cap [0,1]\setminus \mathbb{M}\cap [0,1 ]\cong M\!K^1\setminus M^1$ is a prevalent subset of $M\!K^1$. 
\end{proof}

\begin{remark}
The  proof of Theorem \ref{ThmShy}   also establishes the general fact that
$ \mathbb{M}$ is a shy  subset  in $\mathbb{H}$.
\end{remark}

\subsection{General Asset Structures}
The Bachelier model we presented allows for negative values of the price process. Theorem \ref{ThmRad} is still valid, when  our $\mathbb{E}$-Brownian motion $B=B^+ +B^-$ of the Bachelier model is decomposed into  the positive $B^+$ and negative part $B^-$. The trading strategies are then  given by $\theta^{i,+}_t=\theta^i_t1_{\{B_t\geq 0\}}$ and $\theta^{i,-}_t= -\theta^i_t1_{\{B_t< 0\}}$ where $\theta^i$ denotes the fractions invested in the uncertain assets of Theorem \ref{ThmRad}. In the same fashion, as mentioned in Section 5 of \cite{DuffieHuang85}, the number of assets becomes $2\cdot d +1$.

On the other it is also possible to implement the Arrow-Debreu equilibrium with an arbitrary symmetric $\mathbb{E}$--martingale. A canonical example is  a symmetric $\E$--martingale of exponential form: 
$$ M_t=\exp \left( \int_0^t \eta_s \textnormal{d}B_s + \frac{1}{2}\int_0^t \eta^2_s   \textnormal{d}\langle B\rangle_s \right)$$
The strictly positive process $(M_t)$ is the unique solution of the following linear stochastic differential equation with respect to the $\EE$--Brownian motion $B$,\footnote{ See  Section 5 in  \cite{pe10} for details and especially Remark 1.3 therein.} $$ \frac{ \textnormal{d}M_t}{ M_t}=\eta_t  \textnormal{d}B_t,\quad M_0=1.$$
From a general perspective, we aim to enlarge the scope of Theorem \ref{ThmRad} by allowing a large family of symmetric $\EE$-martingales as feasible substitutes of the Bachelier Model.
The following result takes a leaf out of \cite{duffie1986stochastic}, where the notion of ``martingale generator" points to the implementability of an Arrow--Debreu Equilibrium under some classical probability space.
\begin{proposition}
Theorem \ref{ThmRad}. is still valid if we replace the implementing $B$ with a symmetric $\mathbb{E}$--martingale $M_t=M_0+ \int_0^t V_t \textnormal{d}B_t$, such that  $V_t\neq 0$ $\mathcal{P}$--q.s.
and $V_t\in L^\infty_\mathcal{P}$.
\end{proposition} 
\begin{proof}We show that every $\int H \textnormal{d}B $ can be written as a $\int  H^M \textnormal{d}M$.

 By Proposition 3.3 in \cite{soner2011martingale} states that under every $P\in\mathcal{P}$,  the It\^o integral $\int H\textnormal{d}B$ with respect to some $H\in \mathcal{M}$, coincides $P$-almost surely with  the stochastic integral under $(\Omega,\mathbb{H},\EE)$.

By using the representation of $\EE$ from Proposition \ref{PropRep},  we have $\langle M\rangle^P= \int V^2 \textnormal{d}\langle B\rangle^P$ $P$-almost surely, where $\langle B\rangle^P$ denotes the quadratic variation of the $\EE$-Brownian motion $B$ under $P$, since $B$ is an $E^P$-martingale under every $P\in\mathcal{P}$. Since $V$ is bounded, we have $\Vert \eta\Vert_\mathcal{M}<\infty$ if and only if  
\begin{eqnarray*}
  \sup_{P\in\mathcal{P}}E^P\int_0^T \!\!\eta_t^2 \textnormal{d}\langle M \rangle^P_t<\infty  
\end{eqnarray*}
Consequently, a  stochastic integral with respect to the present symmetric $\mathbb{E}$--martingale $M= \int V \textnormal{d}B$ can then be written as
$\int_0^t \theta_s\textnormal{d}M_s= \int_0^t \theta_s\textnormal{d}\int^s_0 V\textnormal{d}B =  \int_0^t \theta_sV_s\textnormal{d}B_s,$
where $\theta\in \mathcal{M}$, see (\ref{Theta}) for the exact of $\mathcal{M}$.  The   equations hold $P$-almost surely for every $P\in\mathcal{P}$. 

From the conditions on $V$ it follows that $M$ may represent payoffs in $\mathbb{M}$ similarly to  $B$, since for every $X\in \mathbb{M}$ (with $\EE X=0$), we have $$X=\int_0^T  \theta_t  \textnormal{d}B_t=\int_0^T \frac{ \theta_t}{V_t}  V_t \textnormal{d}B_t=\int_0^T  \frac{ \theta_t}{V_t}  \textnormal{d}M_t $$
and this yields $H^M=\frac{ \theta_t}{V_t}\in \mathcal{M}$. 
\end{proof}

\begin{remark}
In finance models it is more usual to take a commodity space of  square--integrable random variables, i.e. $\E X^2 <\infty $.   As stated  in the Martingale Representation Theorem \ref{MRT},  square--integrable random variables   are included. If $\mathbb{H}=L^2_\mathcal{P}$  all results of Section 4 remain valid, see Remark \ref{L2}.
\end{remark}

\section{Conclusion}

This paper establishes a crucial difference of risk and Knightian uncertainty. Under risk, dynamic trading of  few long--lived assets 
suffices to implement the efficient allocations of Arrow--Debreu equilibria as dynamic Radner equilibria if the number of traded assets is equal to the number of sources of uncertainty. 
This result generically fails under Knightian uncertainty even in the stylized framework of no aggregate uncertainty and for nominal asset structures.

All results of the paper are formulated in terms of Peng's sub--linear expectation space $(\Omega,\mathbb{H},\E)$. As stated in Proposition \ref{PropRep}, the Knightian expectation $\E$ can be represented by a set priors that corresponds to  different volatility processes $\sigma_t(\omega)$ that live within  constant bounds $\Sigma$. A further extension refers to the possibility to  extend results when the Knightian expectation is  induced by   time--dependent  and stochastic volatility bounds.
 For instance,  \cite{EpsteinJi14} introduce  a more general  family of time--consistent  conditional Knightian  expectations. Since most results of the paper are heavily based on the Martingale representation Theorem \ref{MRT}, extensions to more general volatility specifications crucially  depend on the availability  of an analogous  martingale Representation.
\begin{appendix}

\section{\!\!Knightian Uncertainty in Continuous Time}

We fix a time horizon $T>0$. Our state space consists of all continuous trajectories on the time interval $[0,T]$ that start in zero:
$$ \Omega = C_0^d([0,T])=\left\{\omega : [0,T] \to \mathbb R^d : \omega(0)=0\right\}\,.$$
The coordinate process
$$B_t(\omega)=\omega(t)$$ 
will describe the information flow of our economy as in the usual continuous--time diffusion model. As agents live in a Knightian world, we do not assume that the distribution $P$ of the process $B$ is commonly known. 
Instead, we just use a nonlinear expectation $\EE$ which is defined on a suitably rich space $\mathbb H$ of functions on $\Omega$ in the sense of the following definition.

\begin{definition}\label{DefE}
  Let $\mathbb H$ be a vector lattice of functions from $\Omega$ to $\mathbb R$ that contains the constant functions. We call $(\Omega, \mathbb H, \mathbb E)$ an \textit{uncertainty space} if the mapping
  $$\mathbb E : \mathbb H \to \mathbb R$$ satisfies the following properties:
  \begin{enumerate}
    \item preserves constants: $\mathbb E c = c $ \quad\qquad\qquad\qquad \:for all $c \in \mathbb R$,
\\[-1.7em]
    \item  monotone: \qquad \qquad  $\mathbb  E X \le \EE Y$\quad\quad\quad \quad \qquad for all  $X, Y \in \mathbb H$, $ X \le Y$ 
\\[-1.7em]
        \item sub-additive:\qquad \quad\:\:$\EE (X+Y) \le \EE X + \EE Y \:\textit{ for all } X, Y \in \mathbb H$,
\\[-1.7em]
        \item  homogeneous: \:\quad\:\:$
\quad\EE(\lambda X) = \lambda \EE X \qquad \qquad\textit{for }\lambda>0  \textit{ and }X \in \HH$.
  \end{enumerate}
We call $\EE$ a (Knightian) expectation.
\end{definition}

  We start here with the notion of an uncertainty space rather than modeling the set of priors because we want to stress that one can build a whole new theory of \textit{uncertainty} (rather than \textit{probability} theory) by starting with the notion of an uncertainty space rather than a probability space, as the work of \cite{peng2007g2} demonstrates. Peng calls such spaces sublinear expectation spaces, but from a philosophical point of view, the name ``uncertainty space'' seems quite fitting to us.

\paragraph{Distributions under $(\Omega,\mathbb{H},\mathbb{E})$ and $\mathbb{E}$-Brownian Motion }
Of particular importance to us is the fact that one can develop the notion of Brownian motion in this Knightian setup. 
We assume throughout the paper that $B$ is an $\EE$-Brownian motion\footnote{Again, we slightly deviate from Peng's nomenclature where this object is called $G$--Brownian motion - no $G$ exists at this point.}, the ambiguous version of classic  Brownian motion on Wiener space. 

We would like to take a second to explain how such an ambiguous Brownian motion is to be understood. 
 With regard to the concept  of ambiguity free (in the strong sense)  from Definition \ref{DefNoamb},          we start how one can define the notion of a ``distribution'' of a function $X \in \HH$ in our setting. When $\HH$ is sufficiently rich (what we assume), then for every continuous and bounded function $f: \mathbb R \to \mathbb R$, $f(X) \in \HH$ as well. We can then define a new operator $$\FF^X$$ on the space $C_b(\mathbb R)$ of continuous, bounded real functions by setting
$$\FF^X(f)=\EE f(X)\,.$$ We call the operator $\FF^X$ the (uncertain) distribution of $X$. 
We consequently say that $X,Y\in \HH$ have identical uncertainty, or $ X \stackrel{d}{=} Y$,  if $\FF^X=\FF^Y$. 

The notion of independence is crucial for probability theory. We follow Peng again in letting $Y$ be ($\EE$-)independent of $X$ if for all continuous bounded functions $f: \mathbb R^2 \to \mathbb R$ we have
$$ \EE f(X,Y)= \EE \left.\EE f(x,Y)\right|_{x=X} \,.$$
One can thus  first fix the value of $X=x$, take the expectation with respect to $Y$, and then take the expectation with respect to $X$. This is one way to generalize the notion of independence to the Knightian case. Without going into the philosophical issues involved here, we just take this approach.  In the same vein, we call $Y$ independent of $X_1, \ldots, X_n \in \mathbb H$ if for all continuous bounded functions $f: \mathbb R^{n+1} \to \mathbb R$ we have
$$ \EE f(X_1, \ldots,X_n ,Y)= \EE \left.\EE f(x_1,\ldots,x_n,Y)\right|_{x_1=X_1,\ldots,x_n=X_n} \,.$$ 

The class of (normalized) normal distributions  is infinitely divisible. In particular, if we have two independent standard normal distributions, then for positive numbers $a$ and $b$, the new variable $Z=aX+bY$ is again normally distributed with variance $a^2+b^2$. This  property is well known to characterize the class of normal distributions. One can take this characterization of normal distributions to call $X\in \HH$ $\EE$--normal if for any $Y \in \HH$ which has  identical uncertainty and is independent of $X$, the uncertain variable $Z=aX+bY$ has the same uncertainty as $\sqrt{a^2+b^2} X$.

We have now the tools at hand to define uncertain $\mathbb{E}$-Brownian motion. 
$B$ is called an $\EE$--Brownian motion if all increments are independent of the past and identically $\EE$--normal: for all $s,t \ge 0$ and all $0\le t_1 \le t_2 \le \ldots \le t_n \le t$ the increment $B_{t+s}-B_t$ is independent of $B_{t_1}, \ldots, B_{t_n}$.

It is, of course, a completely nontrivial question whether such an ambiguous Brownian motion exists. This has been shown by Shige Peng (\cite{peng2007g2}) with the help of the theory of viscosity solutions of nonlinear partial differential equations. An alternative route proceeds via the construction of a suitable set of multiple priors. Indeed, readers familiar with the literature on ambiguity aversion in decision theory or the theory of risk measures in mathematical finance might immediately  anticipate a representation of  our Knightian expectation in terms of a set of probability measures. For the case of Knightian Brownian motion, the set of probability measures has a special structure that we now describe.
\subsection{Representing Priors and Volatility Uncertainty}
We start with the simpler one--dimensional case. Fix two bounds $0<\sigmaunten\le \sigmaoben$. The set $\cP^1$ consists of all probability measures $P$ on $\Omega$ endowed with the Borel $\sigma$--field that make $B$ a martingale whose quadratic variation $\qv{B}$ is $P$--almost surely between the following two bounds:
$$ \sigmaunten^2 t \le \qv{B}_t \le \sigmaoben^2 t \,.$$
In general, the set of priors $\cP^d$ can be parametrized by a subset $\Theta$ of $\mathbb R^{d \times d}$; this set describes the possible volatility structures of the $d-$dimensional Knightian Brownian motion.    Theorem 52 in  \cite{DenisHuPeng11} implies the next results.

\begin{proposition}\label{PropRep}
  For any $X\in \mathbb{H}$, we have the representation
$$\EE X = \sup_{P \in \cP^d} E^P X$$ where $E^P$ is the probabilistic expectation of $X$ under the probability measure $P$.   $\cP^d$  is a weakly--compact set, with respect to the topology induced  $C_b(\Omega)$.
\end{proposition}
Wwe give a more detailed description of the lattice $\HH$ of payoffs we are working with  and set for the rest of the appendix $\mathcal{P}=\mathcal{P}^d$. Peng constructs the $\mathbb{E}$--Brownian motion first on the set of all locally Lipschitz functions of $B$ that satisfy a polynomial growth constraint. The space $\HH$ is obtained by closing this space under the norm  
$$ \| X\|_{\infty}=\inf\{    M\geq 0:| X|  \leq M\:\: \mathcal{P}\textit{-q.s.}\},$$
where $\mathcal{P}$-q.s. refers to $P$-almost surely for every $P\in\mathcal{P}$.
As indicated in the introduction and Section 2, we may  formulate the   uniform version of Luisin's property:   A mapping $X:\Omega\rightarrow\mathbb{R}$ is said to be
quasi-continuous (q.c.) if for all $\epsilon
> 0$  there exists an open set $O$ with $c(O)=\sup_{P\in\mathcal{P}}P(O) < \epsilon$ such
that $X|_{\Omega\setminus O}$ is continuous.

Similarly to Lebesgue spaces based on a probability space, we restrict attention to equivalent classes.   Under $\mathbb{E}$, as shown in \cite{DenisHuPeng11},  we have the following  representation of our commodity space
\begin{eqnarray*}
L^\infty_\mathcal{P}=\left\{X\in  \mathcal{L}:X\text{ has a q.c. version and }  \Vert X\Vert_\infty <\infty
\right\}
\end{eqnarray*}
where $\mathcal{L}$ denotes the space of $\mathcal{N}$-equivalence classes   of   measurable payoffs and  $\mathcal{N}:=\{ X\:\: \mathcal{F}\textit{-measurable and }X= 0 \:\:  \mathcal{P}\textit{-q.s.}\}$ are the trivial payoffs with respect to $\mathcal{P}$  that  do not charge any $P\in\mathcal{P}$.  We say that $X$ has a $\mathcal{P}$-\textit{q.c. version}  if
there is  a quasi--continuous function
$Y:\Omega\rightarrow\mathbb{R}$ with $X = Y$ q.s.

\begin{remark}\label{L2}
Instead of $\Vert\cdot\Vert_\infty$, one may take  $\Vert\cdot\Vert_2=(\mathbb{E}\vert X\vert^2)^{\frac{1}{2}}$ and establish with Theorem 25 in \cite{DenisHuPeng11} that  the space $L_\mathcal{P}^2$, the completion of  $C_b(\Omega)$ under $\Vert\cdot\Vert_2$ is given by 
$$L^2_\mathcal{P}=\left\{X\in  \mathcal{L}:X\text{ has a q.c. version, }  \Vert X\Vert_2 <\infty\:, 
\lim_{n\rightarrow\infty}\mathbb{E}\vert X\vert 1_{\{\vert X\vert >n\}}=0
\right\}.$$ 
The results of the remaining appendix holds also under $L_\mathcal{P}^2$, so that Theorem \ref{ThmRad} and Theorem \ref{ThmShy} are still valid.
\end{remark}

\subsection{Conditional Knightian Expectation}
 For the purpose of a martingale representation theorem we need a  well--defined   conditional   expectation. In accordance with the representation of $\mathbb{E}$ in Proposition \ref{PropRep}, we denote for each $P\in \cal P$, the conditional probability by $P_t=P(\cdot\vert \mathcal{F}_t) $. Fix a volatility regime $(\sigma_t)\in [\sigmaunten, \sigmaoben]$ and denote the resulting martingale law by $P^\sigma$. The set of priors with a  time--depending restriction on the related information set $\mathcal{F}_t$, generated by  $(B_s)_{s\leq t}$, is given by
\begin{align*}
\mathcal{P}_{t,\sigma}=\big\{ P\in \mathcal{P} : P_t=P^\sigma_t\textnormal{ on }\mathcal{F}_t\big\}.
\end{align*}
This set of priors consists of all extensions of $P^\sigma_t$ from $\mathcal{F}_t$ to $\mathcal{F}_T$ within $\mathcal{P}$. All priors in $\mathcal{P}_{t,\sigma}$ agree with $P^\sigma$ in the events up to time $t$, as illustrated in Figure 1.
\begin{figure}[!h]
	\centering
		\includegraphics[width=.65\textwidth]{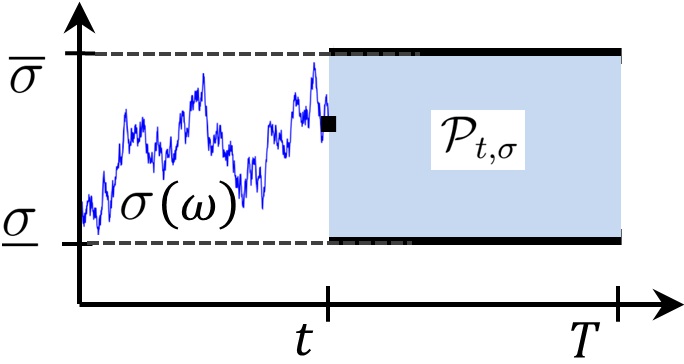}
\\[-.5em]
\caption{The representing priors of a conditional sub--linear expectation } 
\end{figure}
As we  are seeking for a rational--updating principle, we
note, the following formulation of conditiong is closely related to dynamic consistency or rectangularity of \cite{EpsteinSchneider03}. 

 The efficient use of information is commonly formalized  by the concept of  conditional expectations and   depends on the underlying uncertainty model. 
 We introduce a  \textit{universal} conditional expectation, that is under  every prior almost surely equal to the  maximum of  relevant conditional expectations. This concept is formulated in the following. 

Let  $L^2_{t,\mathcal{P}}\subset  L_\mathcal{P}^2$ denote the  subspace of $\mathcal{F}_t$--measurable payoffs.
For all $X\in L_\mathcal{P}^2$ there exists an $\mathcal{F}_t$-measurable random variable $ \mathbb{E}_tX\in L_{t,\mathcal{P}}^2$  such that
\begin{eqnarray*}
\mathbb{E}_tX= \sup_{P\in\mathcal{P}_{t,\sigma}}{E}^{P}_tX,\quad P^\sigma  \textnormal{-a.s. } \textnormal{ for every }
P^\sigma \in \cal P .
\end{eqnarray*} 
 The linear conditional expectation ${E}^{P}_t$ under some  $P$ has strong connections to a positive linear projection operator. In the presence of multiple priors, the conditional updating in an ambiguous environment involves  a sub--linear projection  $\mathbb{E}_t:L^2_\mathcal{P}\rightarrow L^2_{t,\mathcal{P}} $.  
 In this regard  the  conditional Knightian expectations satisfies  a rational--updating principle,  with $\mathbb{E}_0=\mathbb{E}$.  
\begin{lemma}\label{lem}
$(\mathbb{E}_t)$ meet the law of iterated expectation:  $\mathbb{E}_s\circ \mathbb{E}_t=\mathbb{E}_s$, $s\leq t$.
\end{lemma}
Moreover, as shown in  Proposition 16 of \cite{peng2007g2}, every $\EE_t$ satisfies the properties of Definition \ref{DefE} now in the conditional sense, while the property constants preserving extends  to  $\EE_tX=X$ for every $X\in L_{t,\mathcal{P}}^2$.

\subsection{Spanning and  Martingales} 
We proceed similarly to the single prior case, where the  Radner implementation in continuous time is based on a  classical martingale representation. As indicated in Proposition \ref{PropRep}, the multiple prior model enforces  a conditional sub--linear expectation  and  spawns an   elaborated martingale representation. 

We start with a notion of  martingales under the conditional expectation $\mathbb{E}_t$.  
Fix a random variable $X\in L^2_\mathcal{P}$.
As stated in Lemma \ref{lem}, the time consistency of  the  conditional Knightian expectation allows to define a martingale similarly to the  single prior setting,
as being its own estimator.
\begin{definition}
An $(\mathcal{F}_t)$-adapted process $(X_t)$ is an $\mathbb{E}$-\textnormal{martingale} if
\begin{align*}
X_s =  \mathbb{E}_sX_t\quad \mathcal{P}\textit{-q.s.}\quad\textit{for all }s\leq t.
\end{align*}
We call $X$ a \textnormal{symmetric} $\mathbb{E}$-\textnormal{martingale} if $X$ and $-X$ are both  $\mathbb{E}$-martingales.
\end{definition}
The nonlinearity of the conditional expectation implies that if  $(X_t)$ is an $\mathbb{E}$--martingale, then $-X$
is not necessarily an $\mathbb{E}$--martingale. Intuitively, the negation let $\mathbb{E}$  become  super--additive.

We come now to the representation of $\mathbb{E}$-martingales and specify  the space of admissible integrands $\Theta(S)$ taking values in $\mathbb{R}^d$. All processes we consider are $(\mathcal{F}_t)$-progressively measurable.\footnote{The filtration $\mathcal{F}_t$, deviates from the well--known two--step augmentation procedure from the stochastic analysis literature, i.e. including the null--sets and taking the right continuous version. Usually this new filtration is said to satisfy the ``usual conditions". As mentioned  in  Section 2 of \cite{soner2011martingale}, this assumption is no longer required. Consequently, the usually questionable assumption of a too rich information structure at time $0$ can be dropped. }
We begin with the space of well--defined integrands when the Bachelier model builds up the integrator:
\begin{eqnarray}\label{Theta}
\Theta(B)=\Big\{\eta\in \mathcal{M}: \eta_t \textit{ satisfies } (\ref{1}) \Big\},
\end{eqnarray}
where
$ \mathcal{M}$ is closure of  piecewise constant progressively measurable processes $\sum_{k\geq 0} \eta_{t_k}1_{[t_k,t_{k+1})}$ with  $\eta_{t_k}\in L_\mathcal{P}^2$,  under the norm    $\Vert \eta\Vert_\mathcal{M}= \mathbb{E}\int_0^T\eta_t^2 \textnormal{d}\langle B \rangle_t$, see Remark \ref{L2}

Condition (\ref{1})  states the  self--financing property  of an adapted and $B$-integrable  $\eta$.
 For the arguments in Theorem \ref{ThmRad},  it is crucial  if a payoff $X\in L^2_\mathcal{P}$ can be represented or replicated  in terms of a  stochastic integral.
To formulate Theorem \ref{MRT}, set
\begin{eqnarray}\label{K}
\mathbb{K}=\Big\{
(K_t):  K_0=0,  \textit{cont. paths  } \mathcal{P}\textit{-q.s.}, {increasing, }\: \mathbb{E}\sup_{t\in[0,T]} K_t^2 <\infty \Big\}.
\end{eqnarray}
 The following Theorem clarifies this issue, see \cite{soner2011martingale} for a  proof.
\begin{theorem} \label{MRT}
 For every  $X\in L^2_\mathcal{P}$, there exist a  unique pair $(\eta,K)\in\mathcal{M}\times \mathbb{K}$, where 
 $(-K_t)$ is a $\mathbb{E}$--martingale, 
such that  for all  $t\in[0,T]$
\begin{eqnarray*}
 \mathbb{E}_tX=\mathbb{E}_0X+   {\int_0^t} \eta_s \textnormal{d}B_s -K_t,\quad \mathcal{P}\textit{-q.s.}
\end{eqnarray*}
\end{theorem}
The  increasing $\E$--martingale  $-K$ refers a correction term for the ``overshooting" of the sub--linear expectation $\mathbb{E}_t$. Specifically, the conditional Knightian  expectation enforces $t\mapsto \mathbb{E}_tX$ to be a supermartingale under every $P\in\mathbb{P}$. 
For some effective priors $P\in \mathcal{P}$,   $\mathbb{E}_tX$  is an $E^P$-martingale. Foreclosing Corollary \ref{MRTCor}, these two possible cases, can be distinguished, by the fact  ${E}^PK_T=0$ if and only if  $E^PX=\mathbb{E}X$. 

 The following corollary   illustrates  which random variables have the replication property in terms of a stochastic integral.  In this connection, a fair game against nature refers to the symmetric $\mathbb{E}$--martingale property. Apparently, in this situation  the process is equivalently an $E^P$-martingale under every $P\in \mathcal{P}$.
\begin{corollary}\label{MRTCor}
The  space $\mathbb{M}$ of  mean ambiguous--free contingent claims is a closed subspace of $\mathbb{H}=L^2_\mathcal{P}$. More precisely,  we have
\begin{eqnarray*}
\mathbb{M}=\Big\{X\in\mathbb{H}:\:\:X=\mathbb{E}X+\int_0^T \eta_s \textnormal{d}B_s \textit{ for some }\eta\in\mathcal{M}\Big\}.
\end{eqnarray*}
\end{corollary}
The notion of perfect replication is associated to the situation when $K\equiv0$. Elements in $\mathbb{M}$ generate symmetric martingales, via the successive application of the conditional  Knightian expectation along the augmented filtration $(\mathcal{F}_t)$. 

\begin{remark}\label{MRTRem}
  When $X$ is contained in a  subset of $ L^2_\mathcal{P}$, including  all $\phi(B_T)$ with $\phi:\mathbb{R}\rightarrow \mathbb{R}$ Lipschitz continuous, then the $\mathbb{E}$--martingale $-K$ admits an explicit  representation, 
\begin{eqnarray*}
K_t=\int_0^t \varphi_r \textnormal{d}\langle B^G\rangle_r - \int_0^t G( \varphi_r) \textnormal{d}r,\quad t\in[0,T],
\end{eqnarray*}
where  $\varphi$ is an endogenous output of the martingale representation, so that $K$ becomes a   function of $\varphi$.  If $d=1$, the function $G$ is given by $G(x)=\frac{1}{2}\sup_{\sigma \in \Sigma} \sigma^2x$. 
As such it is an open problem, if every $X\in L^2_\mathcal{P}$ can be represented in this complete form. We refer to \cite{peng2013complete} for the latest  discussion.
\end{remark}
\end{appendix}

\end{document}